\documentclass[11pt]{article}
\usepackage{amsmath,amsthm}
\usepackage{mathtools}
\mathtoolsset{showonlyrefs}

\usepackage{graphicx,color}

\usepackage{ifxetex,ifluatex}
\ifxetex
\usepackage{xltxtra}
\else
\ifluatex
\usepackage{luatextra}
\else
\usepackage{amssymb}
\usepackage[utf8]{inputenc}
\usepackage[T1]{fontenc}
\usepackage{mathpazo,tgcursor,tgpagella}
\usepackage{textcomp}
\csname else\expandafter\endcsname
\romannumeral -`0% no space till \relax
\fi
\fi
\iftrue\relax%
\defaultfontfeatures{Ligatures=TeX}
\usepackage[math-style=TeX, vargreek-shape=TeX]{unicode-math}
\providecommand{\restriction}{\upharpoonright}
\AtBeginDocument{%
  \let\setminus\smallsetminus% no U+29F5 in TG Pagella Math?
}
\fi

\usepackage[normalem]{ulem}
\usepackage[letterpaper,margin=1in]{geometry}

\usepackage{graphicx,authblk}
\usepackage{color}%% needed for figure CUT3
\usepackage[authoryear]{natbib}
\usepackage[compact]{titlesec} % compresses space between section heads
\usepackage[inline]{enumitem}
\setenumerate[1]{label=(\roman*)}

\usepackage[bookmarksnumbered,bookmarksopen,
unicode]{hyperref}

\DeclareMathOperator{\CLIQUE}{CLIQUE}
\DeclareMathOperator{\rank}{rank}
\newcommand{\nnegrk}{\rank_{+}} % nonnegative rank
\DeclareMathOperator{\vertexset}{vert}
\newcommand{\RR}{\mathbb{R}}

\newcommand{\R}{\RR}

\newcommand {\abs}[1]{\left|#1\right|}
\newcommand {\SDP}{\mathbb S_+}
\newcommand {\SYM}{\mathbb S}

\DeclareMathOperator{\diag}{diag}
\DeclareMathOperator{\xc}{xc}
\newcommand{\xcp}{\xc_{SDP}}

\DeclareMathOperator{\CUT}{CUT}
\DeclareMathOperator{\COR}{COR}

\DeclareMathOperator{\CUTCONE}{CUT-CONE}
\DeclareMathOperator{\CORCONE}{COR-CONE}

\newcommand{\inp}[2]{\langle #1,#2 \rangle}

\newcommand{\onorm}[1]{\|#1\|_{1}}

\newcommand{\conv}[1]{\operatorname{conv}\left(#1\right)}
\newcommand{\rec}[1]{\operatorname{rec}\left(#1\right)}
\newcommand{\cone}[1]{\operatorname{cone}\left(#1\right)}
\newcommand {\card}[1]{\left|#1\right|}
\DeclareMathOperator{\Row}{Row}
\DeclareMathOperator{\Col}{Col}
\newcommand {\binSet}{\{0,1\}}

\newcommand*{\set}[2]{\left\{#1\,\middle|\,#2\right\}}

\newcommand*{\probability}[1]{\operatorname{\mathbb{P}}\left[#1\right]}
\newcommand*{\probProv}[2]{\operatorname{\mathbb{P}}%
  \left[#1\,\middle|\,#2\right]}
\newcommand*{\expectedValue}[1]{\operatorname{\mathbb{E}}\left[#1\right]}
\newcommand*{\expectProv}[2]{\operatorname{\mathbb{E}}%
  \left[#1\,\middle|\,#2\right]}
\newcommand*{\entropy}[1]{H\left(#1\right)}
\newcommand*{\entropyProv}[2]{H\left(#1\,\middle|\,#2\right)}

\DeclareMathOperator{\ErowBig}{row-big}
\DeclareMathOperator{\EcolBig}{column-big}
\DeclareMathOperator{\Esmall}{small} % \small is a LaTeX command

\newcommand*{\size}[1]{\left|#1\right|}
\newcommand*{\maxNorm}[1]{\left\|#1\right\|_{\infty}}

\newtheorem{thm}{Theorem}
\newtheorem{lem}[thm]{Lemma}
\newtheorem{prop}[thm]{Proposition}

\theoremstyle{definition}
\newtheorem{ex}{Example}

\hypersetup{
  pdftitle={Approximation Limits of Linear Programs (Beyond Hierarchies)},
  pdfauthor={Gábor Braun, Samuel Fiorini,
    Sebastian Pokutta, David Steurer},
  pdfkeywords={},% should be same as argument of \keywords
  pdfsubject={MSC2000 Primary ;
    Secondary}} % should be "MSC2000" + argument of \subjclass{}

\title{Approximation Limits of Linear Programs\\ (Beyond Hierarchies)}
\author[1]{Gábor Braun}
\affil[1]{Universität Leipzig,
 Institut für Informatik,
 PF 100920,
 04009 Leipzig,
 Germany.
 \textit{Email:}~gabor.braun@informatik.uni-leipzig.de}

\author[2]{Samuel Fiorini}
\affil[2]{Department of Mathematics, 
  Université libre de Bruxelles CP 216, 
  Bd. du Triomphe, 
  1050 Brussels, 
  Belgium. 
  \textit{Email:}~sfiorini@ulb.ac.be}

\author[3]{Sebastian Pokutta}
\affil[3]{ISyE, Georgia Institute of Technology,
  Atlanta, GA,
  USA.
  \textit{Email:}~sebastian.pokutta@isye.gatech.edu}

\author[4]{David Steurer}
\affil[4]{Department of Computer Science,
Cornell University,
Ithaca, NY 14853,
United States.
\textit{Email:}~dsteurer@cs.princeton.edu}

\begin{document}

\maketitle

\begin{abstract}
We develop a framework for proving approximation limits 
of polynomial-size linear programs from lower bounds on
the nonnegative ranks of suitably defined matrices. This 
framework yields unconditional impossibility results that are applicable to 
\emph{any} linear program as opposed to only programs generated by
hierarchies. Using our framework, we prove that 
$O(n^{1/2-\epsilon})$-approximations for CLIQUE 
require linear programs of size~$\cramped{2^{\cramped{n^{\Omega(\epsilon)}}}}$.
This lower bound applies to linear programs using a certain encoding of
CLIQUE as a linear optimization problem. Moreover, we establish a similar 
result for approximations of semidefinite programs by linear programs.  

Our main technical ingredient is a quantitative improvement of Razborov's 
rectangle corruption lemma (1992) for the high error regime, which gives 
strong lower bounds on the nonnegative rank of shifts of the unique 
disjointness matrix.
\end{abstract}

%\redTodo{This section is a mixture of introduction and preliminaries.  It
%  would be better to make this introduction-only, i.e., define only
%  the notions necessary to understand the introduction, and move
%  technical details, precise definitions, corner cases, useful tricks
%  to the framework section, e.g., to the preliminaries part.}

\section{Introduction}

%\redTodo{Reorganization in progress (I have removed all definitions
%for now, to see only the "substance").}

%\subsection{Context}
  
Linear programs (LPs) play a central role in the design of 
approximation algorithms, see, e.g., \citep{VaziraniBook01,%
WilliamsonShmoysBook11,LauRaviSinghBook11}. Therefore, 
understanding the limitations of LPs as tools for designing
approximation algorithms is an important question. 

The first generation of results studied the limitations of
\emph{specific} LPs by seeking to determine their integrality 
gaps. The second generation of results, pioneered by 
\cite{AroraBollobasLovasz02}, studied the limitations of 
\emph{structured} LPs such as those generated by 
\emph{lift-and-project procedures} or \emph{hierarchies} 
(e.g., \cite{SheraliAdams1990} and \cite{LovaszSchrijver1991}).
%See the \hyperref[sec:previous-work]{previous work section} below
%for a more detailed account of the relevant literature.

In this work, we start a third generation of results that 
apply to \emph{any} LP for a given problem. For example, 
our lower bounds address the following question: Is there a 
polynomial-size linear programming relaxation $\mathsf{LP}_n$ 
for CLIQUE that achieves a $n^{\Theta(1)}$-approximation for
all graphs with at most $n$ vertices? We develop a framework
for reducing questions of this kind to lower bounds on the 
nonnegative rank\footnote{The \emph{nonnegative rank} of 
a matrix $M$, denoted $\nnegrk(M)$, is the minimum $r$ such 
that $M = TU$ where $T$ and $U$ are nonnegative matrices with 
$r$ columns and $r$ rows, respectively.} of certain matrices 
associated to the problem, and then prove lower bounds for 
the matrices corresponding to CLIQUE.

The matrices studied here are related to the unique disjointness
problem, a variant of the famous disjointness problem from 
communication complexity (see, e.g., \cite{ChattopadhyayPitassi10}
for a survey). In the \emph{disjointness problem} (DISJ), both Alice 
and Bob receive  a subset of~$[n] := \{1,\ldots,n\}$. They have to 
determine whether the two subsets are disjoint. The \emph{unique 
disjointness problem} (UDISJ) is the promise version of the 
disjointness problem where the two subsets are guaranteed to have 
at most one element in common. Denoting the binary encoding of the 
sets of Alice and Bob by $a, b \in \{0,1\}^n$, respectively, this 
amounts to computing the Boolean function
\(
\text{UDISJ}(a,b) \coloneqq 1-a^{\intercal}b
\) 
on the set of pairs $(a,b) \in \{0,1\}^n \times \{0,1\}^n$ with 
$a^{\intercal} b \in \{0,1\}$. Viewing it as a partial $2^n \times 2^n$ 
matrix, we call $\text{UDISJ}$ the \emph{unique disjointness matrix}.

It is known that the communication complexity of UDISJ is $\Omega(n)$
bits for deterministic, nondeterministic and even randomized
communication protocols%
~\citep{KalyanasundaramSchnitger92,Razborov92,BarYossef2004702}.
One consequence of this is that the nonnegative rank of \emph{any} 
matrix obtained from UDISJ by filling arbitrarily the blank entries
(for pairs $(a,b)$ with $a^{\intercal} b > 1$) and perhaps adding 
rows and/or columns is still $2^{\Omega(n)}$. Indeed: (i) the support of the 
resulting matrix has $\Omega(n)$ nondeterministic communication complexity
because it contains UDISJ, (ii) for every matrix $M$, $\log \nnegrk(M)$ is 
lower bounded by the nondeterministic communication complexity of (the 
support matrix of) $M$~\citep{Yannakakis91}.

In a recent paper \cite{extform4} proved strong lower bounds on the
size of LPs expressing the traveling salesman problem (TSP), or more
precisely on the size of extended formulations of the TSP polytope
(see Section~\ref{sec:framework} for definitions of concepts 
related to polyhedra, extended formulations and slack matrices). 
Their proof works by embedding UDISJ in a slack matrix of the TSP 
polytope of the complete graph on $\Theta(n^2)$ vertices. This solved 
a question left open in \cite{Yannakakis91}. We use a similar approach 
for approximate extended formulations. In case of CLIQUE, our approach 
requires lower bounds on the nonnegative rank of partial matrices 
obtained from the UDISJ matrix by adding a positive offset to all 
the entries.

\subsection{Related Work}

Our results are closely related to previous work in communication
complexity for the (unique) disjointness problem and related problems. 
Lower bounds of \(\Omega(n)\) on the randomized, bounded error 
communication complexity of disjointness were established in
\cite{KalyanasundaramSchnitger92}. In \cite{Razborov92} the 
distributional complexity of unique disjointness problem was 
analyzed, which in particular implies the result of 
\cite{KalyanasundaramSchnitger92}. In that famous paper, 
Razborov proved the following \emph{rectangle corruption lemma}: 
for every large rectangle within UDISJ, the number of \(0\)-entries 
is proportional to the number of \(1\)-entries.
 
The most recent proof that the 
randomized, bounded error communication complexity of DISJ
is $\Omega(n)$ is due to \cite{BarYossef2004702} and is 
based on information theoretic arguments. This leads to a lower bound for
randomized communication within a high-error regime, that is, when the
error probability is close to $1/2$. Here we derive a strong generalization
dealing with shifts for approximate EFs and we recover the
high-error regime bound.

There has been extensive work on LP and SDP hierarchies/relaxations 
and their limitations; we will be only able to list a few here. In
\cite{charikar2009integrality}, strong lower bounds (of \(2-\epsilon\)) 
on the integrality gap for \(n^{\epsilon}\) rounds of the  Sherali-Adams 
hierarchy when applied to (natural relaxations of) VERTEX COVER, Max~CUT, 
SPARSEST CUT have been been established via embeddings into \(\ell_2\); see 
also \cite{Charikar2010} for limits and tradeoffs in metric embeddings. For
integrality gaps of relaxations for the KNAPSACK problem see \cite{Karlin2011}. 
A nice overview of the differences and similarities of the Sherali-Adams, 
the Lovász-Schrijver and the Lasserre hierarchies/relaxations can be found 
in \cite{laurent2003comparison}. 

Similar to the level of a hierarchy, we have the notion of \emph{rank} 
for the Lovász-Schrijver relaxation and rank correspond to a similar 
complexity measure as the level. 
% While the level of a hierarchy typically relates to the
% degree of considered monomials, the Lovász-Schrijver relaxation is a
% particular operation that can be iterated however. 
The rank is the minimum number of application of the Lovász-Schrijver 
operator \(N\) until we obtain the integral hull of the polytope under
consideration. Rank lower bounds of \(n\) for Lovász-Schrijver relaxations 
of CLIQUE have been obtained in \cite{CD}; a similar result for Sherali-Adams 
hierarchy can be found in \cite{laurent2003comparison}. 

In \cite{singh2010improving} integrality gaps, after adding few 
rounds of Chvátal-Gomory cuts, have been studied for problems including
\(k\)-CSP, Max~CUT, VERTEX COVER, and UNIQUE LABEL COVER showing that in
some cases (e.g., \(k\)-CSP) the gap can be significantly reduced whereas 
in most other cases the gap remains high. 

In the context of SDP relaxations, in 
particular formulations derived from the Lovász-Schrijver \(N_+\) hierarchies (see
\cite{LovaszSchrijver1991}) and the
Lasserre hierarchies (see \cite{Lasserre02}) there has been
significant work in recent years. For example, 
\cite{arora2009expander} obtained a \(O(\sqrt{\log n})\) upper bound 
on a suitable SDP relaxation of SPARSEST CUT. For lower
bounds in terms of rank, see e.g., \cite{schoenebeck2008linear} for
the \(k\)-CSP in the Lasserre hierarchy or
\cite{schoenebeck2007linear} for VERTEX COVER in the semidefinite
Lovász-Schrijver hierarchy.
Motivated by the Unique Games Conjecture, several works studied upper and
lower bounds for SDP hierarchy relaxations of Unique Games (see for example,
\cite{GuruswamiS11b,barak2011rounding,BarakGHMRS12,BarakBHKSZ12}).
%In \cite{extform4} a characterization of semidefinite EFs
%via one-way quantum communication complexity is established.

Approximate extended formulations have been studied before, for 
specific problems, e.g., KNAPSACK in \cite{bienstock2008approximate},
or as a general tool, see \cite{vyve2006approximate}.

For recent results on computing the nonnegative rank see, e.g., 
\cite{arora2012nnr}.

\subsection{Contribution}

The contribution of the present paper is threefold.

\begin{enumerate}
\item We develop a framework for proving lower bounds on
  the sizes of approximate EFs. Through a generalization of
  Yannakakis's factorization theorem, we characterize the minimum size
  of a $\rho$-approximate extended formulations as the nonnegative rank 
  of any slack matrix of a \emph{pair} of nested polyhedra.
  Thus we reduce the task of proving
  approximation limits for LPs to the task of obtaining lower bounds
  on the nonnegative ranks of associated matrices. Typically,
  these matrices have no zeros, which renders it impossible to use
  nondeterministic communication complexity. We emphasize the fact
  that the results obtained within our framework are unconditional. In
  particular, they do not rely on P $\neq$ NP.

\item We extend Razborov's rectangle corruption lemma to deal 
 with shifts of the UDISJ matrix. As a consequence, we prove 
 that the nonnegative rank of any matrix obtained from the UDISJ
 matrix by adding a constant offset to every entry is still 
 $2^{\Omega(n)}$. Moreover, we show that the nonnegative rank is 
 still $2^{\Omega(n^{2\epsilon})}$ when the offset is at
 most $n^{1/2-\epsilon}$.
 To our knowledge, these are
 the first strong lower bounds on the nonnegative rank of
 matrices that contain no zeros. %(Furthermore, the relative difference
 %between any two entries of some of our shifted UDISJ matrices is tiny.) 
 Our extension of Razborov's lemma allow us to recover known lower
 bounds for DISJ in the high-error regime of \cite{BarYossef2004702}.

\item We obtain a strong hardness result for CLIQUE w.r.t.\ 
   a natural linear encoding of the problem. From the results
  described above, we prove that the size of every
  $O(n^{1/2-\epsilon})$-approximate EF for CLIQUE is
  $2^{\Omega(n^{2\epsilon})}$. %We see this as the first step in
  %obtaining lower bounds on the sizes of approximate EFs for (faithful
  %linear encodings of) other problems. -> not sure what is meant here!
  Finally, we observe that the
  same bounds hold for approximations of SDPs by LPs. This suggests
  that SDP-based approximation algorithms can be significantly stronger than
  LP-based approximation algorithms. The inapproximability of SDPs by
  LPs has some interesting consequences. In particular we cannot
  expect to convert SDP-based approximation algorithms into LP-based
  ones by approximating the PSD-cone via linear programming. 
\end{enumerate}

We point out that our framework readily generalizes to SDPs
by replacing nonnegative rank with PSD rank (see
\cite{GouveiaParriloThomas2011} for a definition of
the PSD rank). However, no strong bound on PSD rank 
seems to be currently in sight.

Finally, we report that the results of this paper have inspired 
further research. 
\begin{itemize}
\item \cite{BM13} improved our lower bound on the nonnegative rank of shifted 
UDISJ matrices and obtain super-polynomial lower bounds for shifts up to 
$O(n^{1-\epsilon})$,
hence matching the algorithmic hardness of approximation
for CLIQUE. This was achieved by pioneering information-theoretic methods for 
proving lower bounds on the nonnegative rank. An alternative information 
theoretic approach for lower bounding the nonnegative rank which simplifies 
and slightly improves the results in \cite{BM13} has been presented in
\cite{BP2013commInfo}. 
This last paper also establishes that matrices obtained 
from shifts of UDISJ by removing rows and columns, or flipping entries,
still have high nonnegative rank.

\item \cite{CLRS13} obtain lower bounds on the size of LPs approximating 
Max~CSP. In particular, they prove that approximating Max~CUT (with
nonnegative weights) with a constant factor less than $2$ requires
$n^{\Omega(\log n / \log \log n)}$. This solves a conjecture we
stated in an earlier version of this text.

\item \cite{Rothvoss14} proved a $2^{\Omega(n)}$ lower bound on the 
nonnegative rank of the slack matrix of the perfect matching polytope
by a significant modification of Razborov's lemma. This exciting result
essentially proves that there are is no small LP that can solve all 
weighted instance of the matching problem on a $n$-vertex complete graph.
\end{itemize}

\subsection{Outline}

We begin in Section \ref{sec:framework} by setting up our framework for 
studying approximate extended formulations of combinatorial optimization 
problems. Then we extend Razborov's rectangle corruption lemma in Section 
\ref{sec:push-razb-result} and use this to prove strong lower bounds on 
the nonnegative rank of shifts of the UDISJ matrix. Finally, we draw 
consequences for CLIQUE and approximations of SDPs by LPs in 
Section~\ref{sec:consequences}.

\section{Framework for Approximation Limits of LPs}
\label{sec:framework}

In this section we establish our framework for studying
approximation limits of LPs. First, we define in details the
concepts of linear encodings and approximate extended formulations.
Second, we prove a factorization theorem for pairs of
nested polyhedra reducing existential questions on approximate
extended formulations to the computation of nonnegative ranks of 
corresponding slack matrices.

\subsection{Preliminaries}
\label{sec:preliminaries}

A (convex) \emph{polyhedron} is a set $P \subseteq \RR^d$ that 
is the intersection of a finite collection of closed halfspaces.
In other words, $P$ is a polyhedron if and only if $P$ is the set of 
solutions of a finite system of linear inequalities and possibly equalities.
(Note that every equality can be represented by a pair of inequalities.)
Equivalently, a set $P \subseteq \RR^d$ is a polyhedron if and only
if $P$ is the Minkowski sum of the convex hull $\conv{V}$ of a finite 
set $V$ of points and the conical hull $\cone{R}$ of a finite set $R$
of vectors, that is, $P = \conv{V} + \cone{R}$.

Let $P \subseteq \RR^d$ be a polyhedron. The \emph{dimension} of \(P\) 
is the dimension of its affine hull \(\operatorname{aff}(P)\). A \emph{face} 
of~$P$ is a subset $F \coloneqq \{x \in P \mid w^{\intercal} x = \delta\}$ 
such that \(P\) satisfies the inequality $w^{\intercal} x \leqslant 
\delta$. Note that face $F$ is again a polyhedron.
%A face is called \emph{proper} if it is not the polyhedron itself.
A \emph{vertex} is a face of dimension $0$, i.e., a point.
A \emph{facet} is a face of dimension one less than \(P\).
The inequality $w^{\intercal} x \leqslant \delta$ is called 
\emph{facet-defining} if the face $F$ it defines is a facet.
The \emph{recession cone} $\rec{P}$ of $P$ is the set of directions $v \in \RR^d$ such that,
for a point $p$ in $P$, all points $p + \lambda v$ where $\lambda \geqslant 0$
belong to $P$. The recession cone of $P$ does not depend on the base point $p$, 
and is again a polyhedron (even more, it is a polyhedral cone). The elements
of the recession cone are sometimes called \emph{rays}.

A \emph{(convex) polytope} $P \subseteq \RR^d$ is a bounded polyhedron. 
Equivalently, $P$ is a polytope if and only if $P$ is the convex hull 
$\conv{V}$ of a finite set $V$ of points. 
Let $P \subseteq \RR^d $ be a polytope. 
Every (finite or infinite) set $V$ such that $P = \conv{V}$
contains all the vertices of $P$.
Letting $\vertexset(P)$ denote the vertex set of $P$,
then we have $P = \conv{\vertexset(P)}$. 
Every (finite) system describing $P$ contains
all the facet-defining inequalities of $P$,
up to scaling by positive numbers and
adding equalities satisfied by \emph{all} points of $P$.
Conversely, a linear description of $P$ can be obtained by
picking one defining inequality per facet and
adding a system of equalities describing 
\(\operatorname{aff}(P)\). A \emph{$0/1$-polytope} in 
$\RR^d$ is simply the convex hull of a subset of $\{0,1\}^d$.

For more about convex polytopes and polyhedra,
see the standard reference~\cite{Ziegler}.

\subsection{Linear Encodings of Problems and Approximate EFs}

A \emph{linear encoding} of a (combinatorial optimization) problem 
is a pair $(\mathcal{L},\mathcal{O})$ where $\mathcal{L} \subseteq
\{0,1\}^*$ is the set of \emph{feasible solutions} to the problem and
$\mathcal{O} \subseteq \RR^*$ is the set of \emph{admissible 
objective functions}. An \emph{instance} of the linear encoding is 
a pair $(d,w)$ where $d$ is a positive integer and $w \in \mathcal{O} 
\cap \RR^d$. Solving the instance $(d,w)$ means finding $x \in 
\mathcal{L} \cap \{0,1\}^d$ such that $w^{\intercal} x$ is either 
maximum or minimum, according to the type of problem under consideration.

\begin{ex}[Linear encoding of metric TSP]
\label{ex:metric_TSP}
In the natural linear encoding of the metric traveling salesman 
problem (metric TSP), the feasible solutions $x \in \mathcal{L}$ are 
the characteristic vectors (or incidence vectors) of tours of the 
complete graph over $[n]$ for some $n \geqslant 3$, and the admissible 
objective functions $w \in \mathcal{O}$ are all nonnegative vectors 
$w = (w_{ij})$ such that $w_{ik} \leqslant w_{ij} + w_{jk}$ for all 
distinct $i$, $j$ and $k$ in $[n]$. All vectors are encoded in $\RR^d$, 
where $d = \binom{n}{2}$. By considering all possible $n \geqslant 3$, we 
obtain the pair $(\mathcal{L},\mathcal{O})$ corresponding to metric TSP. 
(Recall that metric TSP is a minimization problem.)
\end{ex}

For every fixed dimension $d$, a linear encoding $(\mathcal{L},\mathcal{O})$ naturally defines a pair of nested convex sets $P \subseteq Q$ where
\begin{align}
P &\coloneqq \conv{\{x \in \{0,1\}^d \mid x \in\mathcal{L}\}}, \quad \text{and}\\
Q &\coloneqq \{x \in \RR^d \mid \forall w \in \mathcal{O} \cap \RR^d :
w^{\intercal} x \leqslant \max \{w^{\intercal} z \mid z \in P\}\}
\end{align}
if the goal is to maximize and
\(
Q \coloneqq \{x \in \RR^d \mid \forall w \in \mathcal{O} \cap \RR^d :
w^{\intercal} x \geqslant \min \{w^{\intercal} z \mid z \in P\}\}
\) 
if the goal is to minimize. Intuitively, the vertices of $P$ encode the feasible solutions 
of the problem under consideration and the defining inequalities of 
$Q$ encode the admissible objective functions.  Notice that $P$
is always a 0/1-polytope but $Q$ might be unbounded and, in some
pathological cases, nonpolyhedral. Below, we will mostly consider 
the case where $Q$ is polyhedral, that is, defined by a finite number
of ``interesting'' inequalities.

Given a linear encoding $(\mathcal{L},\mathcal{O})$ 
of a maximization problem, and $\rho \geqslant 1$,
a \emph{$\rho$-approximate
extended formulation} (EF) is an extended formulation 
$Ex + Fy = g$, $y \geqslant \mathbf{0}$ with $(x,y) \in \RR^{d+r}$
such that 
\begin{align}
\max \{w^{\intercal} x \mid Ex + Fy = g,\ y \geqslant \mathbf{0}\}
&\geqslant \max \{w^{\intercal} x \mid x \in P\} \quad \text{for all} \quad w \in \RR^d \quad \text{and}\\
\max \{w^{\intercal} x \mid Ex + Fy = g,\ y \geqslant \mathbf{0}\}
&\leqslant \rho \max \{w^{\intercal} x \mid x \in P\} \quad \text{for all} \quad w \in \mathcal{O} \cap \RR^d. 
\end{align}
Letting 
$K \coloneqq \{x \in \RR^d \mid \exists y \in \RR^r:
Ex + Fy = g,\ y \geqslant \mathbf{0}\}$, we 
see that this is equivalent to $P \subseteq K \subseteq \rho Q$. For a 
minimization problem, we require 
\begin{align}
\min \{w^{\intercal} x \mid Ex + Fy = g,\ y \geqslant \mathbf{0}\}
&\leqslant \min \{w^{\intercal} x \mid x \in P\}
\quad \text{for all} \quad w \in \RR^d \quad \text{and}\\ 
\min \{w^{\intercal} x \mid Ex + Fy = g,\ y \geqslant \mathbf{0}\}
&\geqslant \rho^{-1} \min \{w^{\intercal} x \mid x \in P\}
\quad \text{for all} \quad w \in \mathcal{O} \cap \RR^d. 
\end{align}
This is equivalent to $P \subseteq K \subseteq \rho^{-1} Q$.

\begin{ex}[Approximate extended formulation of metric TSP]
We return to Example \ref{ex:metric_TSP}. It is known that the 
Held-Karp relaxation \(K\) of the metric TSP has integrality gap 
at most $3/2$ (see \cite{heldKarp70}, \cite{Wolsey80}). In geometric 
terms, this means that $P \subseteq K \subseteq 2/3 \cdot Q$.
Although $K$ is defined by an exponential number of inequalities,
it is known that it can be reformulated with a polynomial number of 
constraints by adding a polynomial number of variables, see, e.g., 
\cite{Carr_et_al09}. That is, the Held-Karp relaxation $K$ has a 
polynomial-size extended formulation. Thus, the pair $(\mathcal{L},
\mathcal{O})$ for the metric TSP has a polynomial-size $3/2$-approximate
EF.
\end{ex}

We require the following \emph{faithfulness condition}: every instance 
of the problem can be mapped to an instance of the linear encoding 
in such a way that feasible solutions to an instance of the problem 
can be converted in polynomial time to feasible solutions to the 
corresponding instance of the linear encoding without deteriorating 
their objective function values, and vice-versa. Roughly speaking,
we ask that each instance of the problem can be encoded as an 
instance of the linear encoding.

For linear encoding of graph problems, such as the maximum clique 
problem (CLIQUE), the set of feasible solutions is not allowed to 
depend on the input graph, which therefore must be encoded solely in 
the objective function. The set of feasible solutions is only allowed 
to depend on the size $n$ of the ground set.

\begin{ex}[Max~$k$-SAT]
Consider the maximum $k$-SAT problem 
(Max~$k$-SAT), where $k$ is constant. Letting $u_1$, \ldots, 
$u_n$ denote the variables of a Max $k$-SAT instance, we encode 
the problem in dimension $d = \Theta(n^k)$. For each nonempty 
clause $C$ of size at most $k$, we introduce a variable $x_C$.
Collectively, these variables define a point $x \in \RR^d$. Given 
a truth assignment, we set $x_C$ to $1$ if $C$ is satisfied and 
otherwise we set $x_C$ to $0$. Letting $n$ vary, this defines a 
language $\mathcal{L} \subseteq \{0,1\}^*$. We let 
$\mathcal{O} \coloneqq \{0,1\}^*$.

The pair $(\mathcal{L},\mathcal{O})$ defines a linear encoding 
of Max~$k$-SAT because each instance of Max~$k$-SAT can be
encoded as an instance of $(\mathcal{L},\mathcal{O})$. More
precisely, to any given set of clauses over $n$ variables, we 
can associate a dimension $d = \Theta(n^k)$ and weight vector
$w \in \{0,1\}^d$ such that maximizing $\sum w_C x_C$ for 
$x \in \mathcal{L} \cap \{0,1\}^d$ corresponds to finding a 
truth assignment that maximizes the number of satisfied clauses.

Finally, we remark that the EF defined by the inequalities 
$0 \leqslant x_C \leqslant 1$ and $x_C \leqslant \sum_{u_i \in C} 
x_{\{u_i\}} + \sum_{\bar{u}_i \in C} (1-x_{\{u_i\}})$ for all 
clauses $C$ is a polynomial-size $4/3$-approximate EF for 
Max~$k$-SAT, as follows from \cite{Goemans94anew}. 
\end{ex}

\subsection{Factorization Theorem for Pairs of Nested Polyhedra}

Let $P$ and $Q$ be polyhedra with $P \subseteq Q
\subseteq \RR^d$. An \emph{extended formulation} (EF) \emph{of 
the pair} $P,Q$ is a system $Ex + Fy = g$, $y \geqslant
\mathbf{0}$ defining a polyhedron $K \coloneqq \{x \in \RR^d \mid
Ex + Fy = g,\ y \geqslant \mathbf{0}\}$ such that 
$P \subseteq K \subseteq Q$. We denote by $\xc(P,Q)$ 
the minimum size of an EF of the pair $P,Q$. 

Now consider an inner description
\(P \coloneqq \conv{\{v_{1}, \dotsc, v_{n}\}}
+ \cone{\{r_{1}, \dots, r_{k}\}}\) of $P$
and an outer description
$Q \coloneqq \{x \in \RR^d \mid A x \leqslant b\}$ of $Q$,
where the system \(Ax \leqslant b\) consists of $m$ inequalities:
\(A_{1} x \leqslant b_{1}, \dotsc, A_{m} x \leqslant b_{m}\).  The
\emph{slack matrix of the pair} $P,Q$ w.r.t.\ these inner and outer
descriptions is the $m \times (n + k)$ matrix
\(S^{P,Q} =
\left[
\begin{smallmatrix}
  S^{P,Q}_{\mathrm{vertex}} & S^{P,Q}_{\mathrm{ray}}
\end{smallmatrix}
\right]
\)
given by block decomposition into a vertex and ray part:
\begin{align*}
 S^{P,Q}_{\mathrm{vertex}}(i,j) &\coloneqq b_{i} - A_{i} v_{j},
 & i \in [m],\ j \in [n], \\
 S^{P,Q}_{\mathrm{ray}}(i,j) &\coloneqq - A_{i} r_{j},
 & i \in [m],\ j \in [k].
\end{align*}

A \emph{rank-$r$ nonnegative factorization} of an $m \times n$ 
matrix $M$ is a decomposition of $M$ as a product $M = TU$
of nonnegative matrices $T$ and $U$ of sizes $m \times r$ and 
$r \times n$, respectively. The \emph{nonnegative rank} 
$\nnegrk(M)$ of $M$ is the minimum rank $r$ of nonnegative 
factorizations of $M$. In case $M$ is zero, we let $\nnegrk(M) = 0$.
It is quite useful to notice that the 
nonnegative rank of $M$ is also the minimum number of 
nonnegative rank-$1$ matrices whose sum is $M$. From this, 
we see immediately that the nonnegative rank of $M$ is at
least the nonnegative rank of any of its submatrices.

Our first result gives an essentially exact characterization of $\xc(P,Q)$ in 
terms of the nonnegative rank of the slack matrix of the pair $P,Q$.
It states that the minimum extension complexity $\xc(P,Q)$ of a polyhedron 
sandwiched between $P$ and $Q$ equals the nonnegative rank of $S^{P,Q}$
(minus $1$, in some cases). The result readily generalizes Yannakakis's 
factorization theorem~\citep{Yannakakis91}, which concerns the case $P = Q$. 
The idea of considering a pair \(P,Q\) as we do here first appeared 
in \cite{Pashkovich12} and similar ideas appeared earlier in 
\cite{GillisGlineur10}. 

\begin{thm}
\label{thm:sandwich} 
With the above notations, we have
\(\nnegrk(S^{P,Q}) - 1 \leqslant \xc(P, Q) \leqslant \nnegrk(S^{P,Q})\)
for every slack matrix of the pair \(P, Q\). If the affine hull of $P$
is not contained in $Q$ and $\rec{Q}$ is not full-dimensional, we have 
\(\xc(P, Q) = \nnegrk(S^{P,Q})\). In particular, this holds when $P$ 
and $Q$ are polytopes of dimension at least $1$.
\begin{proof}
  First, we deal with degenerate cases. Observe that $\xc(P,Q) = 0$ if and only if
there exists an affine subspace containing $P$ and contained in $Q$, that is, if and
only if the affine hull of $P$ is contained in $Q$. In this case, we have $\nnegrk(S^{P,Q})
\in \{0,1\}$, so the theorem holds. 

  Now assume that the affine hull of $P$ is not contained in $Q$. Then, $\nnegrk(S^{P,Q})
\geqslant 1$ because having $\nnegrk(S^{P,Q}) = 0$ means either that $S^{P,Q}$ is empty,
that is, $m = 0$ or $n+k = 0$, or that $S^{P,Q}$ is the zero matrix. In all cases, this 
contradicts our assumption that the affine hull of $P$ is not contained in $Q$. 

  Next, let $S^{P,Q} = TU$ be any rank-$r$ nonnegative factorization 
of $S^{P,Q}$ with $r = \nnegrk (S^{P,Q}) \geqslant 1$. This factorization
decomposes into blocks:
\(S^{P,Q}_{\mathrm{vertex}} = T U_{\mathrm{vertex}}\)
and \(S^{P,Q}_{\mathrm{ray}} = T U_{\mathrm{ray}}\).
Consider the system
\begin{equation}
\label{eq:sandwich_system}
Ax + Ty = b,\ y \geqslant \mathbf{0}
\end{equation}
and the corresponding polyhedron \(K \coloneqq
\{x \in \RR^d \mid Ax + Ty = b,\ y \geqslant \mathbf{0}\}\).

We verify now that 
\(P \subseteq K \subseteq Q\). The inclusion $K \subseteq Q$ simply 
follows from $Ty \geqslant \mathbf{0}$. For the inclusion 
$P \subseteq K$, pick a vertex \(v_j\) of \(P\)
and observe that $(x,y) =
(v_j,U_{\mathrm{vertex}}^j)$ satisfies \eqref{eq:sandwich_system}, where
$U_{\mathrm{vertex}}^j$ denotes
the $j$th column of $U_{\mathrm{vertex}}$,
because $Av_j + TU_{\mathrm{vertex}}^{j} =
Av_j + b - Av_j = b$ and $U^j \geqslant \mathbf{0}$.
Similarly, for every ray \(r_{j}\) we obtain a ray
\((r_{j}, U_{\mathrm{ray}}^{j})\) of \(K\) as
\(A r_{j} + T U_{\mathrm{ray}}^{j} = 0\) and
\(U_{\mathrm{ray}}^{j} \geqslant \mathbf{0}\).

Thus we obtain that 
\eqref{eq:sandwich_system} is a size-$r$ EF of the pair $P,Q$. 
Therefore, $\xc(P,Q) \leqslant \nnegrk(S^{P,Q})$.

Finally, suppose that the system
\begin{equation}
\label{eq:sandwich_system2}
Ex + Fy = g,\ y \geqslant \mathbf{0}
\end{equation}
defines a size-$r$ EF of the pair $P, Q$. Let $L \subseteq \RR^{d+r}$ 
denote the polyhedron defined by \eqref{eq:sandwich_system2}, and 
let $K \subseteq \RR^d$ denote the orthogonal projection of $L$ into 
$x$-space. 

Since $P \subseteq K$, for each point $v_j$, there
exists $w_{j} \in \RR_{+}^{r}$ such that $(v_{j}, w_{j}) \in L$.
Similarly, for each ray \(r_{j}\) there exists a \(z_{j} \in \RR_{+}^{r}\)
with \((r_{j}, z_{j})\) a ray of \(L\).
Let \(W\) be the matrix with columns \(w_{j}\),
and \(Z\) be the matrix with columns \(z_{j}\).

Since $K \subseteq Q$,
by Farkas's lemma,
$Ax \leqslant b$ can be derived from \eqref{eq:sandwich_system2},
i.e., there exists a matrix \(T\) and a vector
\(c \geqslant \mathbf{0}\) with \(A = T E\), \(b = T g + c\)
and \(T F \geqslant 0\).
This gives the factorizations
% b - A v_j = T g + c - T E v_j = T (g -  E v_j) + c = T F w_j + c
\(S^{P,Q}_{\mathrm{vertex}} = (T F) W + c \mathbf{1}^\intercal\)
% - A r_j = - T E r_j = T F z_j
and \(S^{P, Q}_{\mathrm{ray}} = (T F) Z\),
resulting in the rank-$(r+1)$ nonnegative factorization
\(S^{P, Q} =
\left[
  \begin{smallmatrix}
    T F & c
  \end{smallmatrix}
\right]
\cdot
\left[
  \begin{smallmatrix}
    W & Z \\
    \mathbf{1}^\intercal & \mathbf{0}^\intercal
  \end{smallmatrix}
\right].
\)
Taking $r = \xc(P,Q)$, we find \(\nnegrk(S^{P,Q}) \leqslant \xc(P,Q) + 1\).

Finally, when \(\rec{Q}\) is not full-dimensional, then \(c\) above can be chosen 
to be \(\mathbf{0}\). This simplifies the factorization, and yields the sharper 
inequality \(\nnegrk(S^{P,Q}) \leqslant \xc(P,Q)\).
\end{proof}
\end{thm}

Let \(P,Q\) be as above and $\rho \geqslant 1$. Then $\rho Q = 
\{x \in \RR^d \mid Ax \leqslant \rho b\}$ and the slack matrix 
of the pair $P,\rho Q$ is related to the slack matrix of the pair 
$P,Q$ in the following way:

\begin{gather*}
  S^{P,\rho Q}_{\mathrm{vertex}}(i,j) = \rho b_{i} - A_{i} v_{j}
  = (\rho - 1) b_{i} + b_{i} - A_{i} v_{j}
  = S^{P,Q}_{ij} + (\rho-1) b_{i}, \\
  S^{P,\rho Q}_{\mathrm{ray}}(i,j)
  = S^{P,Q}_{ij}.
\end{gather*}

Theorem \ref{thm:sandwich} directly yields the following result.  

\begin{thm}
\label{thm:framework}
Consider a maximization problem with a linear encoding.
Let $P, Q \subseteq \RR^d$ be the pair of polyhedra
associated with the linear encoding, and let 
$\rho \geqslant 1$. Consider any slack matrix $S^{P,Q}$ for the 
pair $P, Q$ and the corresponding slack matrix $S^{P,\rho Q}$ for 
the pair $P, \rho Q$. Then the minimum size of a $\rho$-approximate EF
of the problem, w.r.t.\ the considered linear encoding, is $\nnegrk (S^{P,\rho Q}) +
\Theta(1)$, where the constant is $0$ or $1$. For a minimization problem, the minimum 
size of a $\rho$-approximate EF is $\nnegrk (S^{P,\rho^{-1}Q}) + \Theta(1)$. 
\end{thm}

Fixing $\rho \geqslant 1$, Theorem \ref{thm:framework} characterizes
the minimum number of inequalities in any LP providing a 
$\rho$-approximation for the problem under consideration. 
We point out that the theorem directly generalizes to SDPs, 
by replacing nonnegative rank by PSD rank~\citep{GouveiaParriloThomas2011}. Here, we focus 
on LPs and nonnegative rank. As a matter of fact, strong lower 
bounds on the PSD rank seem to be currently lacking.

\subsection{A Problem with no Polynomial-Size Approximate EF}
\label{sec:case-correlation-cut}

We conclude this section with an example
showing the necessity to restrict the set of admissible
objective functions rather than allowing every $w \in \RR^*$
(that is $P = Q$).

Let $K_n = (V_n,E_n)$ denote the $n$-vertex complete graph. 
For a set $X$ of vertices of $K_n$, we let $\delta(X)$ denote 
the set of edges of $K_n$ with one endpoint in $X$ and the other 
in its complement $\bar{X}$. This set $\delta(X)$ is known as the 
\emph{cut} defined by $X$. For a subset $F$ of edges of $K_n$, we 
let $\chi^{F} \in \mathbb{R}^{E_n}$ denote the \emph{characteristic vector}
(or \emph{incidence vector}) of $F$, with $\chi^F_e = 1$ if $e \in F$ and 
$\chi^F_e = 0$ otherwise. The \emph{cut polytope} $\CUT(n)$ is defined as 
the convex hull of the characteristic vectors of all cuts in the complete 
graph $K_n = (V_n,E_n)$. That is,
\[
\CUT(n) \coloneqq \conv{\{\chi^{\delta(X)} \in \mathbb{R}^{E_n} \mid X \subseteq V_n\}}.
\]
A related object is the \emph{cut cone}, defined as the cone generated by 
the \emph{cut-vectors} $\chi^{\delta(X)}$:
\[
\CUTCONE(n) \coloneqq \cone{\{\chi^{\delta(X)} \in \mathbb{R}^{E_n} 
\mid X \subseteq V_n\}}.
\]

Consider the maximum cut problem (Max~CUT) with 
\emph{arbitrary} weights, and its usual linear encoding.
With this encoding we have $P = Q = \CUT(n)$.
Our next result states that this problem has no $\rho$-approximate
EF, whatever $\rho \geqslant 1$ is. Intuitively, this phenomenon 
stems from the fact that, because $\mathbf{0}$ is a vertex of the cut
polytope, every approximate EF necessarily
\lq{}captures\rq{} all facets of the cut polytope incident to $\mathbf{0}$
(see Figure~\ref{fig:CUT3}). These facets define the cut cone, which turns
out to have high extension complexity. Although this follows rather easily
from ideas of \cite{extform4}, we include a proof here for completeness. 

\begin{figure}[ht]
\centering
\input{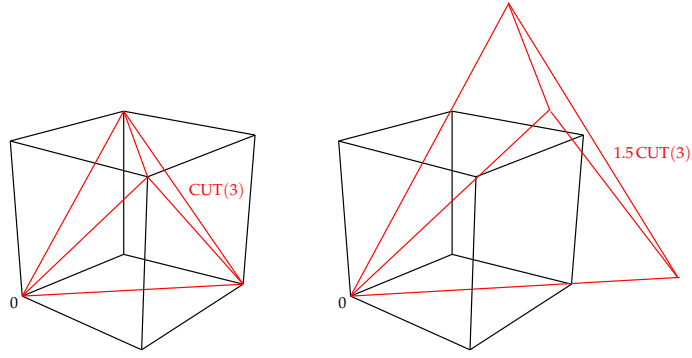}
\caption{$\CUT(3)$ and a dilate $\rho \CUT(3)$ for $\rho = 1.5$.}
\label{fig:CUT3}
\end{figure}

\begin{prop}
\label{prop:MAXCUT_arbitrary}
For every $\rho \geqslant 1$, every $\rho$-approximate EF of the
Max~CUT problem with arbitrary weights has size $2^{\Omega(n)}$.
More precisely, disregarding the value of $\rho \geqslant 1$, we
have $\xc(\CUT(n),\rho \CUT(n)) = 2^{\Omega(n)}$.
\begin{proof}
Let $Ex + Fy = g$, $y \geqslant \mathbf{0}$ denote a minimum size
$\rho$-approximate EF of $\CUT(n)$. We claim that
\begin{equation}
\label{eq:CUTCONE_EF}
Ex + Fy - \lambda g = \mathbf{0},\ y \geqslant \mathbf{0},\ \lambda \geqslant 0
\end{equation}
is an EF of the cut cone. Let $K$ be the polyhedron obtained by projecting 
the set of solutions of \eqref{eq:CUTCONE_EF} into $x$-space. Clearly, $K$ 
is a cone containing all the cut-vectors $\chi^{\delta(X)}$, from which we get
that $\CUTCONE(n) \subseteq K$. Now take any point $(x,y,\lambda)$ 
satisfying \eqref{eq:CUTCONE_EF}. If $\lambda = 0$ then necessarily
$x = \mathbf{0}$ because $Ex + Fy = \mathbf{0}$,
$y \geqslant \mathbf{0}$ defines the recession
cone of a polyhedron that projects into $\rho\CUT(n)$, which is bounded.
In this case we have $x = \mathbf{0} \in \CUTCONE(n)$.
Assume that $\lambda > 0$.
Then $E\lambda^{-1}x + F\lambda^{-1}y = g$ and
$\lambda^{-1}y \geqslant \mathbf{0}$
which implies that $\lambda^{-1}x$ is in $\rho \CUT(n)$. Thus $\rho^{-1}
\lambda^{-1}x$ is in $\CUT(n)$ and $x$ is thus a positive combination of
cut-vectors, hence $x \in \CUTCONE(n)$. This yields $K \subseteq \CUTCONE(n)$.
In conclusion, $K = \CUTCONE(n)$ and \eqref{eq:CUTCONE_EF} is an EF of the
cut cone. The size of this EF is at most $r + 1$,
where $r$ denotes the
size of the given $\rho$-approximate EF of $\CUT(n)$.
Thus $\xc(\CUTCONE(n)) \leqslant r +1$.

By using the correlation mapping (see \cite[p.\ 55]{DezaLaurent}), the cut cone 
has the same extension complexity as its corresponding \emph{correlation cone}, defined as
\[
\CORCONE(n-1)  \coloneqq \cone{\set{\binom{b_0}{b}\binom{b_0}{b}^\intercal}{b_0 \in \binSet, b \in \binSet^{n-2}}}.
\]
We claim that the unique disjointness matrix on $[n-2]$ can be embedded in a slack 
matrix of $\CORCONE(n-1)$. To prove this, consider the $(n-1) \times (n-1)$ rank-$1$ 
positive semidefinite matrices
\begin{equation}
\label{eq:PSD_factorization}
T_a \coloneqq \binom{-1}{a}\binom{-1}{a}^\intercal
\qquad
\text{and}
\qquad
U^b \coloneqq \binom{1}{b}\binom{1}{b}^\intercal
\end{equation}
where $a, b \in \{0,1\}^{n-2}$. The Frobenius inner product $\inp{T_a}{z} \geqslant 0$
of $T_a$ with any correlation matrix $z=\binom{b_0}{b}\binom{b_0}{b}^\intercal$ is 
nonnegative because both matrices are positive semidefinite. Thus $\inp{T_a}{z} \geqslant 0$ 
is valid for all points $z \in \CORCONE(n-1)$, for all $a \in \{0,1\}^{n-2}$. Moreover, 
$\inp{T_a}{U^b} = (1 - a^\intercal b)^2$ for all $a, b \in \{0,1\}^{n-2}$ and thus 
$\inp{T_a}{U^b} = \text{UDISJ}(a,b)$ provided $a^\intercal b \in \{0,1\}$. 

From what precedes, the slack of correlation matrix $U^b$ with respect to the valid inequality 
$\inp{T_a}{z} \geqslant 0$ is $\text{UDISJ}(a,b)$ provided $a^\intercal b \in \{0,1\}$. 
Therefore, $\CORCONE(n-1)$ has a slack matrix that contains UDISJ on $[n-2]$. Because the 
nonnegative rank of any matrix containing UDISJ is $2^{\Omega(n)}$ (this follows from 
\citep{Razborov92}, see \cite[Theorem 1]{extform4}), we conclude that the nonnegative 
rank of some slack matrix of $\CORCONE(n-1)$ is $2^{\Omega(n)}$.
From Theorem
\ref{thm:sandwich} applied to $P = Q = \CORCONE(n-1)$, it follows that $\xc(\CORCONE(n-1)) 
= 2^{\Omega(n)}$. Thus we get 
\[
r + 1 \geqslant \xc(\CUTCONE(n)) = \xc(\CORCONE(n-1)) = 2^{\Omega(n)}
,
\]
from which we obtain $r = 2^{\Omega(n)}$.
The result then follows immediately.
\end{proof}
\end{prop}

\section{Extension of Razborov's Lemma and Shifts of Unique Disjointness}
\label{sec:push-razb-result}

In the first subsection we generalize Razborov's famous lemma
on the disjointness problem
(see \cite{Razborov92} or \citet[Lemma~4.49]{KushilevitzNisan97}
for the original version).
In the next subsection
we apply it to shift the UDISJ matrix
without significantly decreasing its nonnegative rank,
which will be used in later sections to obtain
lower bounds on approximate extended formulations.

The main improvements to Razborov's lemma are threefold:
\begin{enumerate*}
\item the dependence on the error parameter $\epsilon$
  is made explicit;
\item better analytical estimations are employed
  to improve overall strength of the statement; 
\item probabilities are generalized to expected values
  to homogenize the proof and yield a stronger lemma.
\end{enumerate*}

\subsection{Extension of Razborov's Rectangle Corruption Lemma}
\label{sec:Razborov-lemma}

Suppose that \(n \equiv 3 \pmod{4}\) and let 
\begin{align*}
  \ell &\coloneqq \frac{n+1}{4},\\
  A &\coloneqq \{(a,b) \in 2^{[n]} \times 2^{[n]} \mid |a| = |b| = \ell,\ \card{a \cap b} = 0\},\\
  B &\coloneqq \{(a,b) \in 2^{[n]} \times 2^{[n]} \mid |a| = |b| = \ell,\ \card{a \cap b} = 1\}.
\end{align*}
Thus $A$ is the set of \emph{disjoint} pairs of $\ell$-subsets and 
$B$ is the set of \emph{barely intersecting} pairs of $\ell$-subsets.
Furthermore, let \(\mu\) be any distribution on pairs \((a,b)\) of 
subsets of \([n]\) that is supported on $A \cup B$ and uniform when
conditioned to either $A$ or $B$.

\begin{lem}
\label{lem:Razborov_lemma}
Let $n$, $\ell$, $A$, $B$ and $\mu$ be as above. For every nonnegative 
functions $f$ and $g$ defined on $2^{[n]} \times 2^{[n]}$
we introduce a random variable \(X \coloneqq f(a)g(b)\).
Then for every \(0 < \epsilon < 1\):
\begin{equation}
  \label{eq:Razborov_lemma}
   (1 - \epsilon) \expectProv{X}{A} - \expectProv{X}{B} \leqslant
   \maxNorm{X \restriction (A \cup B)} 2^{-\frac{\epsilon^{2}}{16 \ln 2} \ell + O(\log \ell)},
\end{equation}
where the constant in the \(O(\log \ell)\) is absolute, and $X \restriction (A \cup B)$ denotes
the restriction of $X$ to $A \cup B$.
\end{lem}

Let us write \(I_{C}\) for the indicator of an event \(C\).
In case $f$ and $g$ are both binary, $X$ is the indicator
of a rectangle
$R$, that is $X = I_R$, and \eqref{eq:Razborov_lemma} becomes
\[
(1 - \epsilon) \probProv{R}{A} - \probProv{R}{B} 
\leqslant 2^{-\frac{\epsilon^{2}}{16 \ln 2} \ell + O(\log \ell)},
\]
which is a strengthened version of Razborov's original lemma.

For concreteness, the reader
might find it helpful to imagine that $X$ is the indicator of a 
rectangle in the proof below. Our proof is inspired by the version in 
\citet[Lemma~4.49]{KushilevitzNisan97} and we adopt similar notations.

\begin{proof}[Proof of Lemma~\ref{lem:Razborov_lemma}] The proof is in four main steps.\medskip

\noindent \textsl{Step 1: Expressing $\expectProv{X}{A}$ and
  $\expectProv{X}{B}$ in an alternative framework.}
The statement of the lemma does not depend on
the actual probabilities of \(A\) and \(B\),
hence for convenience, we fix them as
\[
\probability{A} = \frac{3}{4} \qquad \text{and} \qquad
\probability{B} = \frac{1}{4}.
\]
This brings the advantage of the following
alternative description of \(\mu\).

Let \(T = (T_1,T_2,\{i\})\) be a uniformly chosen partition
of \([n]\) into two subsets \(T_{1}\), \(T_{2}\) with \(2 \ell - 1\)
elements each and one singleton \(\{i\}\). Given \(T\) we choose $a$ as 
a uniform $\ell$-subset of $T_1 \cup \{i\} = [n] \setminus T_2$ 
and $b$ as a uniform $\ell$-subset of $T_2 \cup \{i\} = [n] 
\setminus T_1$, independently. This defines a distribution 
$\mu$ that is supported on $A \cup B$, uniform when conditioned 
to either $A$ or $B$ and satisfies
$\probability[T]{B} = \probability[T]{i \in a, i \in b}
= \probability[T]{i \in a} \probability[T]{i \in b} = (1/2)^2 = 1/4$
and thus $\probability[T]{A} = 1 - 1/4 = 3/4$.
In particular,
\(\probability{A} = 3/4\) and \(\probability{B} = 1/4\),
as required.

We begin by rewriting $\expectProv{X}{B}$ and then $\expectProv{X}{A}$ in terms
of the following functions of $T$:
\begin{align}
  \label{eq:Row}
  \Row_{0}(T) &\coloneqq \expectProv{f(a)}{T, i \notin a}, &
  \Row_{1}(T) &\coloneqq \expectProv{f(a)}{T, i \in a}, \\
  \label{eq:Col}
  \Col_{0}(T) &\coloneqq \expectProv{g(b)}{T, i \notin b}, &
  \Col_{1}(T) &\coloneqq \expectProv{g(b)}{T, i \in b}.
\end{align}

We note the following nice interpretation of \(\Row_{0}(T) + \Row_{1}(T)\) and 
\(\Col_{0}(T) + \Col_{1}(T)\), that we will use at the end of the 
proof:
\begin{align}
  \label{eq:21}
  \expectProv{f(a)}{T} &
  \begin{aligned}[t]
    &=
    \underbrace{\expectProv{f(a)}{T, i \in a}}_{\Row_{1}(T)}
    \cdot
    \underbrace{\probProv{i \in a}{T}}_{1/2}
    +
    \underbrace{\expectProv{f(a)}{T, i \notin a}}_{\Row_{0}(T)}
    \cdot
    \underbrace{\probProv{i \notin a}{T}}_{1/2}
    \\
    &= \frac{\Row_{0}(T) + \Row_{1}(T)}{2},
  \end{aligned}
  \\
  \label{eq:22}
  \expectProv{g(b)}{T} &= \frac{\Col_{0}(T) + \Col_{1}(T)}{2}.
\end{align}

Note that:
\begin{enumerate*}
\item
  the distribution of $(a,b)$ conditioned on a given $T$ is a product
distribution (this local independence property is the main reason why we reinterpret the
distribution $\mu$);
\item
  the marginal distributions of $a$ conditioned on
$(T, i \in a, i \in b)$ and $(T, i \in a)$ are the same (and similarly for $b$, we can
remove the condition $i \in a$).
\end{enumerate*}
From these facts, we get
\begin{equation}
\label{eq:expectation_B}
\begin{aligned}
\expectProv{X}{B} 
&= \expectProv{f(a)g(b)}{i \in a, i \in b}\\
&= \expectedValue{\expectProv{f(a)g(b)}{T, i \in a, i \in b}}\\
&= \expectedValue{\expectProv{f(a)}{T, i \in a, i \in b} \expectProv{g(b)}{T, i \in a, i \in b}}\\
&= \expectedValue{\expectProv{f(a)}{T, i \in a} \expectProv{g(b)}{T, i \in b}}\\
&= \expectedValue{\Row_1(T) \Col_1(T)}.
\end{aligned}
\end{equation}
By similar arguments, we find
\begin{align}
\expectProv{X}{A} 
&= \frac{1}{3} \expectProv{f(a)g(b)}{i \notin a, i \notin b} + 
\frac{1}{3} \expectProv{f(a)g(b)}{i \in a, i \notin b} +
\frac{1}{3} \expectProv{f(a)g(b)}{i \notin a, i \in b}\\ 
&= \frac{1}{3} \expectedValue{\Row_0(T) \Col_0(T)} + 
\frac{1}{3} \expectedValue{\Row_1(T) \Col_0(T)} +
\frac{1}{3} \expectedValue{\Row_0(T) \Col_1(T)}.
\end{align}

Pick a $(2\ell-1)$-subset $T_2$ of $[n]$,
that we consider fixed for the time being.
The marginal distributions of $a$ conditioned on the events $T_2$, $(T_2, i \in a)$ 
and $(T_2, i \notin a)$ are the same, namely, the uniform distribution on the 
$\ell$-subsets of $[n] \setminus T_2$. (Note that fixing $T_2$ does not fix $i$, 
which could be any element of $[n] \setminus T_2$.) Thus, we get
\begin{equation}
\label{eq:expectProv_f(a)}
\expectProv{f(a)}{T_2, i \notin a} = \expectProv{f(a)}{T_2, i \in a} = \expectProv{f(a)}{T_2}.
\end{equation}
On the other hand, we have
\begin{equation}
\label{eq:expectProv_Row0(T)}
\begin{aligned}
\expectProv{\Row_0(T)}{T_2}
& = \expectProv{\expectProv{f(a)}{T, i \notin a}}{T_2}\\
& = \expectProv{\frac{\expectProv{f(a) I_{i \notin a}}{T}}{\probProv{i \notin a}{T}}}{T_2}\\
& = 2 \expectProv{f(a) I_{i \notin a}}{T_2}\\
& = \expectProv{f(a)}{T_2, i \notin a}
\end{aligned}
\end{equation}
and similarly
\[
\expectProv{\Row_1(T)}{T_2} = \expectProv{f(a)}{T_2, i \in a}.
\]
From \eqref{eq:expectProv_f(a)}, we conclude
\begin{equation}
\label{eq:expectProv_Row0(1)_Row1(T)_equal}
\expectProv{ \Row_{0}(T) }{T_{2}} = \expectProv{ \Row_{1}(T) }{T_{2}}.
\end{equation}
Therefore (letting $T_2$ vary), 
\begin{align}
\expectedValue{\Row_1(T) \Col_0(T)}
& = \expectedValue{\expectProv{\Row_1(T) \Col_0(T)}{T_2}}\\
& = \expectedValue{\expectProv{\Row_1(T)}{T_2} \Col_0(T)}\\
& = \expectedValue{\expectProv{\Row_0(T)}{T_2} \Col_0(T)}\\
& = \expectedValue{\expectProv{\Row_0(T)\Col_0(T)}{T_2}}\\
& = \expectedValue{\Row_0(T) \Col_0(T)}.
\end{align}
The second and fourth equalities above are due to the fact that $\Col_0(T)$ is constant when $T_2$ is fixed. This is because $\Col_0(T) = \expectProv{g(b)}{T, i \notin b}$ depends only on $T_2$, as the marginal distribution of \(b\) given \((T, i \notin b)\) is uniform on the \(\ell\)-subsets of $T_2$.

Exchanging the roles of rows and columns, we have
\[
\expectedValue{\Row_1(T) \Col_0(T)} = \expectedValue{\Row_0(T) \Col_0(T)}.
\]
In conclusion, we find the following simple expression for $\expectProv{X}{A}$:
\begin{equation}
\label{eq:expectation_A}
\expectProv{X}{A} = \expectedValue{\Row_0(T) \Col_0(T)}.
\end{equation}
\medskip

\noindent \textsl{Step 2: Estimation of $\expectProv{X}{A} - \expectProv{X}{B}$.} 
Via obvious estimates:
\begin{equation}
\label{eq:rectangle_estimation}
\begin{aligned}
\MoveEqLeft
\Row_0(T) \Col_0(T) - \Row_1(T) \Col_1(T)\\
&\leqslant \Row_0(T) \Col_0(T) - \min\{\Row_{0}(T), \Row_{1}(T)\}
\cdot \min\{\Col_{0}(T), \Col_{1}(T)\}\\
&=
\begin{aligned}[t]
  &\Row_{0}(T) (\Col_{0}(T) - \min\{\Col_{0}(t), \Col_{1}(T)\})
  \\
  &
  + (\Row_{0}(T) - \min\{\Row_{0}(t), \Row_1(T)\})
  \min\{\Col_{0}(t), \Col_1(T)
\end{aligned}
\\
&\leqslant \Row_0(T) |\Col_0(T) - \Col_1(T)| + |\Row_0(T) - \Row_1(T)| \Col_0(T).
\end{aligned}
\end{equation}
This argument is depicted on Figure~\ref{fig:rectangle_trick}.

\begin{figure}[ht]
\centering
\input{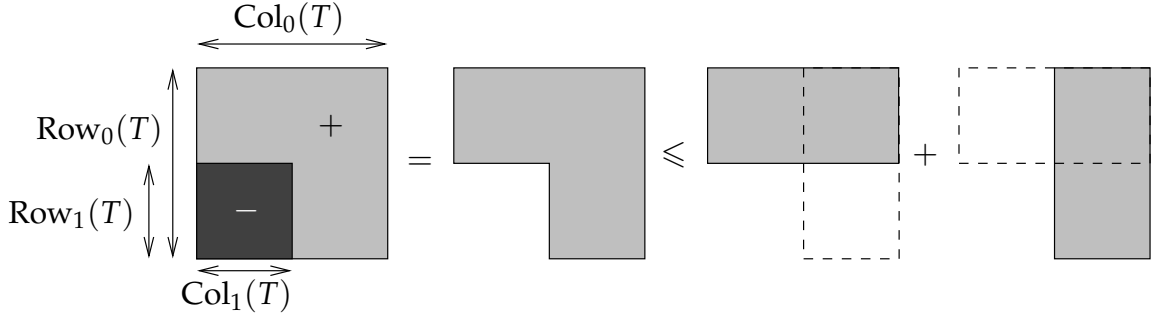}
\caption{The estimation of $\Row_0(T) \Col_0(T) - \Row_1(T) \Col_1(T)$.}
\label{fig:rectangle_trick}
\end{figure}

In Step 3 below,
we will define two events, $\ErowBig(T)$ and $\EcolBig(T)$.
The event $\Esmall(T)$ holds if and only if not both of
$\ErowBig(T)$ and $\EcolBig(T)$ hold.  Thus
\begin{equation}
\label{eq:Epartition}
1 = I_{\ErowBig(T) \cap \EcolBig(T)} + I_{\Esmall(T)}.
\end{equation}
From \eqref{eq:rectangle_estimation}, 
\begin{align}
&(\Row_0(T) \Col_0(T) - \Row_1(T) \Col_1(T)) \cdot I_{\ErowBig(T) \cap \EcolBig(T)}\\
&\leqslant (\Row_0(T) |\Col_0(T) - \Col_1(T)| + |\Row_0(T) - \Row_1(T)| \Col_0(T)) \cdot I_{\ErowBig(T) \cap \EcolBig(T)}\\
&\leqslant \Row_0(T) |\Col_0(T) - \Col_1(T)| \cdot I_{\EcolBig(T)} + |\Row_0(T) - \Row_1(T)| \Col_0(T) \cdot I_{\ErowBig(T)}.
\end{align}
Moreover, we obviously have
\[
(\Row_0(T) \Col_0(T) - \Row_1(T) \Col_1(T)) \cdot I_{\Esmall(T)} \leqslant \Row_0(T) \Col_0(T) \cdot I_{\Esmall(T)}.
\]
Below, we will prove 
\begin{align}
\label{eq:onesided_estimation1}
\expectedValue{\Row_0(T) |\Col_0(T) - \Col_1(T)| \cdot I_{\EcolBig(T)}}
&\leqslant \frac{\epsilon}{2} \expectedValue{\Row_0(T) \Col_0(T)},\\
\label{eq:onesided_estimation2}
\expectedValue{|\Row_0(T) - \Row_1(T)| \Col_0(T) \cdot I_{\ErowBig(T)}}
&\leqslant \frac{\epsilon}{2} \expectedValue{\Row_0(T) \Col_0(T)}, \quad \text{and}\\
\label{eq:small_estimation}
\expectedValue{\Row_0(T) \Col_0(T) \cdot I_{\Esmall(T)}}
&\leqslant \maxNorm{X \restriction (A \cup B)} 2^{-\frac{\epsilon^2}{16 \ln 2} - O(\log \ell)}
\end{align}
By \eqref{eq:expectation_B}, \eqref{eq:expectation_A} and \eqref{eq:Epartition}, 
these upper bounds imply
\begin{align}
&\expectProv{X}{A} - \expectProv{X}{B}\\
&= \expectedValue{\Row_0(T) \Col_0(T) - \Row_1(T) \Col_1(T)}\\
&= \expectedValue{(\Row_0(T) \Col_0(T) - \Row_1(T) \Col_1(T))\cdot(I_{\ErowBig(T) \cap \EcolBig(T)} + I_{\Esmall(T)})}\\
&\leqslant 2 \frac{\epsilon}{2} \expectedValue{\Row_0(T) \Col_0(T)} + 
\maxNorm{X \restriction (A \cup B)} 2^{-\frac{\epsilon^2}{16 \ln 2} \ell - O(\log \ell)}\\
&= \epsilon \expectProv{X}{A} + 
\maxNorm{X \restriction (A \cup B)} 2^{-\frac{\epsilon^2}{16 \ln 2} \ell - O(\log \ell)}
\end{align}
from which the result clearly follows, by rearranging.\medskip

\noindent \emph{Step 3. One-sided error estimation via entropy argument in the ``big'' case.}
Let \(\delta > 0\) be a constant to be chosen later. Essentially, $\delta$ will be the 
coefficient of \(\ell\) in the exponent. Let \(\ErowBig(T)\) denote the 
event \(\expectProv{f(a)}{T_2} > 2^{- \delta \ell -1} \maxNorm{f \restriction 
\binom{[n] \setminus T_{2}}{\ell}}\) where $f \restriction \binom{[n] \setminus T_{2}}{\ell}$ 
denotes the restriction of $f$ to $\ell$-subsets of $[n] \setminus T_{2}$. The event $\EcolBig(T)$
is defined in a similar way. These events depend only on $T_2$ and $T_1$, respectively.

Let \(T_{2}\) be fixed and assume that $\ErowBig(T)$ holds. In particular \(\expectProv{f(a)}{T_2}\) 
is positive. Because $\binom{2\ell - 1}{\ell-1} = \binom{2\ell - 1}{\ell}$, the 
distribution of $a$ given $T_2$ is the same as the distribution of $a$
given $T$, for every fixed choice of $i$. Thus, we have 
\begin{equation*}
  \expectProv{f(a)}{T} =
  \expectProv{f(a)}{T_{2}} =
  \sum_{\substack{x \subseteq [n] \setminus T_{2} \\ \size{x} = \ell}}
  \frac{1}{\binom{2\ell}{\ell}} f(x)
  .
\end{equation*}
(This holds when $f(a)$ is replaced by any function of $a$.)

We define \(s\) as a random \(\ell\)-subset of \([n] \setminus T_{2}\) with distribution
\[
\probProv{s = x}{T_{2}}
 = \frac{f(x)}{\binom{2\ell}{\ell} \expectProv{f(a)}{T_{2}}}
 = \frac{f(x)}{\sum_{\substack{y \subseteq [n] \setminus T_{2} \\ \size{y} = \ell}} f(y)}
 \leqslant \frac{2^{\delta \ell + 1}}{\binom{2\ell}{\ell}}.
\]
Let us introduce the shorthand notation \(\lambda \coloneqq \probProv{i \in s}{T_2}\).
Then
\begin{equation}
\lambda = \frac{\sum_{\substack{x \subseteq [n] \setminus T_{2}\\ \size{x} = \ell,\ x \ni i}} f(x)}{\sum_{\substack{y \subseteq [n] \setminus T_{2} \\ \size{y} = \ell}} f(y)}
= \frac{\frac{1}{\binom{2\ell}{\ell}}\sum_{\substack{x \subseteq [n] \setminus T_{2}\\ \size{x} = \ell,\ x \ni i}} f(x)}{\frac{1}{\binom{2\ell}{\ell}}\sum_{\substack{y \subseteq [n] \setminus T_{2} \\ \size{y} = \ell}} f(y)}
= \frac{\expectProv{f(a) I_{i \in a}}{T}}{\expectProv{f(a)}{T_2}}.
\end{equation}
Hence,
\begin{align}
\label{eq:8}
&\Row_{1} (T) 
= 2 \expectProv{f(a) I_{i \in a}}{T}
= 2 \expectProv{f(a)}{T_{2}} \cdot \probProv{i \in s}{T_2}
= 2 \lambda \expectProv{f(a)}{T_{2}},\\[1ex]
\label{eq:3}
&\Row_{0} (T) 
= 2 \expectProv{f(a) I_{i \notin a}}{T}
= 2 \expectProv{f(a)}{T_{2}} \cdot \probProv{i \notin s}{T_2}
= 2 (1 - \lambda) \expectProv{f(a)}{T_{2}}.
\end{align} 

We now estimate the entropy of \(s\). On the one hand, by subadditivity of the entropy,
we get the following upperbound on $\entropyProv{s}{T_{2}}$:
\begin{equation}
  \label{eq:7}
  \entropyProv{s}{T_{2}}
  \leqslant
  \sum_{j \in [n] \setminus T_{2}} \entropyProv{I_{j \in s}}{T_{2}}
  = 2 \ell \expectProv{\entropy{\lambda}}{T_{2}}.
\end{equation}
In this last equation, $\entropy{\lambda}$ denotes the binary entropy of \(\lambda\). 
On the other hand, we get a lower bound on $\entropyProv{s}{T_2}$ from 
our upper bound on the distribution of $s$ (which induces ``flatness'' of the distribution):
\begin{equation}
  \label{eq:6}
  \begin{split}
  \entropyProv{s}{T_{2}} &= \sum_{x}
  \probProv{s = x}{T_{2}} \log \frac{1}{\probProv{s = x}{T_{2}}}
  \\
  &\geqslant \sum_{x}
  \probProv{s = x}{T_{2}} \log
  \frac{\binom{2\ell}{\ell}}{2^{\delta \ell + 1}}
  =  \log \frac{\binom{2\ell}{\ell}}{2^{\delta \ell + 1}}
  = 2 \ell
  \left( 1 - \frac{\delta}{2}
    - O \genfrac(){}{}{\log \ell}{\ell} \right).
\end{split}
\end{equation}
This implies
\begin{equation}
  \label{eq:16}
  \frac{\delta}{2} + O \genfrac(){}{}{\log \ell}{\ell}
  \geqslant
  \expectProv{1 - \entropy{\lambda}}{T_{2}}.
\end{equation}

To estimate this expression,
we use the Taylor expansion
of the binary entropy function at \(1/2\):
\begin{gather}
  \label{eq:17}
  1 - \entropy{x} 
%=
%  \frac{1}%
%  {2 \ln 2}
%  \sum_{n=1}^{\infty}
%  \frac{{\left( 1 - 2 x \right)}^{2n}}{n (2n - 1)}
  \geqslant
  \frac{
    {\left( 1 - 2 x \right)}^{2}}%
  {2 \ln 2}.
\end{gather}
Hence \eqref{eq:16} yields
\begin{equation}
  \label{eq:23}
  \frac{\delta}{2} + O \genfrac(){}{}{\log \ell}{\ell}
  \geqslant
  \frac{
    \expectProv{
      {\left(
          1 - 2 \lambda
        \right)}^{2}}{T_{2}}}%
  {2 \ln 2}
  \geqslant
  \frac{
    {\left(
        \expectProv{ \abs{1 - 2 \lambda} }{T_{2}}
      \right)}^{2}}%
  {2 \ln 2}.
\end{equation}

From \eqref{eq:expectProv_f(a)}, \eqref{eq:expectProv_Row0(T)} we have
$\expectProv{f(a)}{T_2} = \expectProv{\Row_0(T)}{T_2}$. Using \eqref{eq:3} 
and \eqref{eq:8}, we derive
\begin{equation}
  \label{eq:9}
  \begin{aligned}
  \expectProv{\abs{\Row_0(T) - \Row_{1}(T)}}{T_{2}}
  &= \expectProv{|2(1-\lambda)\expectProv{f(a)}{T_2} - 2 \lambda \expectProv{f(a)}{T_2}|}{T_2}\\
  &= 2 \expectProv{|1 - 2\lambda|}{T_2} \expectProv{f(a)}{T_2}\\
  & \leqslant 2 \sqrt{\delta'} \expectProv{\Row_0(T)}{T_2}.
  \end{aligned}
\end{equation}
with
\begin{equation}
  \label{eq:12}
  \delta' \coloneqq
  \left(
    \delta + O \genfrac(){}{}{\log \ell}{\ell}
  \right) \ln 2.
\end{equation}

We now globalize to prove \eqref{eq:onesided_estimation2}:
\begin{align}
&\expectedValue{|\Row_0(T) - \Row_1(T)| \Col_0(T) I_{\ErowBig(T)}}\\
&= \expectedValue{\expectProv{|\Row_0(T) - \Row_1(T)| \Col_0(T) I_{\ErowBig(T)}}{T_2}}\\
&= \expectedValue{\expectProv{|\Row_0(T) - \Row_1(T)| I_{\ErowBig(T)}}{T_2}\Col_0(T)}\\
&\leqslant \expectedValue{2 \sqrt{\delta'} \expectProv{\Row_0(T)}{T_2} \Col_0(T)}\\
&= 2 \sqrt{\delta'} \expectedValue{\Row_0(T) \Col_0(T)}
\end{align}
We require \(\frac{\epsilon}{2} = 2 \sqrt{\delta'}\), from which we express \(\delta\) in terms 
of \(\epsilon\) using \eqref{eq:12}:
\begin{gather}
  \label{eq:18}
  \delta =
  \frac{\delta'}{\ln 2} - O \genfrac(){}{}{\log \ell}{\ell}
  =
  \frac{\epsilon^{2}}{16 \ln 2}
  - O \genfrac(){}{}{\log \ell}{\ell}
\end{gather}
This concludes the proof of \eqref{eq:onesided_estimation2}. Equation 
\eqref{eq:onesided_estimation1} follows by exchanging rows and columns.
\medskip

\noindent \emph{Step 4: Error estimation in the ``small'' case.} Suppose that for some
given $T$, $\Esmall(T)$ holds because $\ErowBig(T)$ does not hold (the argument is 
similar in case $\EcolBig(T)$ does not hold). Then, using \eqref{eq:21},
\begin{equation}
\Row_0(T)
\leqslant 
\Row_0(T) + \Row_1(T)
=
2 \expectProv{f(a)}{T}.
\end{equation}
Thus
\begin{align}
\Row_0(T) \Col_0(T) 
&\leqslant 2 \expectProv{f(a)}{T} \cdot \expectProv{g(b)}{T, i \notin b}\\
&\leqslant 2^{- \delta \ell} \maxNorm{f(a) \restriction \binom{[n] \setminus T_2}{\ell}} \cdot \maxNorm{g(b) \restriction \binom{T_2}{\ell}}\\
&\leqslant 2^{-\delta \ell} \maxNorm{f(a)g(b) \restriction (A \cup B)}
\end{align}
This is easily seen to imply \eqref{eq:small_estimation}.
\end{proof}

\subsection{Lower Bounds for Shifts of Unique Disjointness}

Now we apply Lemma~\ref{lem:Razborov_lemma} to show that
the nonnegative rank (and hence the communication complexity
in expectation) of any shifted version of the unique disjointness 
matrix remains high. More precisely, let \(M \in \R_+^{2^n \times 2^n}\); 
for convenience we index the rows and columns with elements in \(\binSet^n\). 
We say that \(M\) is a \emph{\(\rho\)-extension of UDISJ}, if \(M_{ab} = \rho\)
whenever \(\card{a \cap b} = 0\) and \(M_{ab} = \rho - 1\)  whenever
\(\card{a \cap b} = 1\) with \(a,b \in \binSet^n\).
Note that for these pairs \(M\) has exclusively positive entries whenever \(\rho >1\).
For \(\rho = 1\) a nonnegative rank of \(2^{\Omega(n)}\)
was already shown in \cite{extform4} via nondeterministic communication
complexity. We now extend this result for a wide range of \(\rho\)
using Lemma~\ref{lem:Razborov_lemma}. 

\begin{thm}[Nonnegative rank of UDISJ shifts]
\label{thm:highNNRPert}
Let \(M \in \R_+^{2^n \times 2^n}\) be a \(\rho\)-extension of
UDISJ as above.  
\begin{enumerate}
\item\label{item:1}
If \(\rho\) is a fixed constant, then \(\nnegrk(M) = 2^{\Omega(n)}\). 
\item\label{item:2}
If \(\rho = O(n^\beta)\) for some constant \(\beta < 1/2\) then
  $\nnegrk(M) = 2^{\Omega(n^{1 - 2 \beta})}$.
\end{enumerate}
\begin{proof}
Without loss of generality, assume $n \equiv 3 \pmod{4}$. 
Let $r = \nnegrk(M)$. Regarding the $2^n \times 2^n$ matrix $M$ as a 
function from $2^{[n]} \times 2^{[n]}$ to $\R$, we can write $M = 
\sum_{i=1}^r X_i$ where $X_i(a,b) = f_i(a)g_i(b)$ for some nonnegative
functions $f_i$ and $g_i$ defined over $2^{[n]}$. Then
\[
\expectProv{M}{A} = \rho
\quad
\text{and}
\quad
\expectProv{M}{B} = \rho - 1.  
\]
On the other hand, by applying Lemma~\ref{lem:Razborov_lemma}
to each $i \in [r]$ and
summing up all equations we find
\begin{align}
(1-\epsilon) \expectProv{M}{A} - \expectProv{M}{B}
&\leqslant \sum_{i=1}^r \maxNorm{X_i \restriction (A \cup B)} 2^{-\frac{\epsilon^{2}}{16 \ln 2} \ell + O(\log \ell)}\\
&\leqslant r \maxNorm{M \restriction (A \cup B)} 2^{-\frac{\epsilon^{2}}{16 \ln 2} \ell + O(\log \ell)}
\end{align}
where $\ell = \frac{n+1}{4}$ as before. By plugging in the values of $\expectProv{M}{A}$,
$\expectProv{M}{B}$ and $\maxNorm{M \restriction (A \cup B)}$, we get
\begin{align*}
  (1-\epsilon) \rho - \rho + 1
  &\leqslant r \cdot \rho \cdot
  2^{-\frac{\epsilon^{2}}{16 \ln 2} \ell + O(\log \ell)},
  \intertext{which provides the lower bound}
  r &\geqslant \left(
    \frac{1}{\rho} - \epsilon
  \right)
  2^{
    \frac{\epsilon^{2}}{16 \ln 2} \ell - O(\log \ell)}.
\end{align*}

If \(\rho\) is constant, this last expression is $2^{\Omega(n)}$ provided
\(\epsilon\) is chosen sufficiently close to \(0\). This proves part 
\ref{item:1} of the theorem.

If \(\rho \leqslant C n^{\beta}\) for some positive constant $C$,
then we can take $\epsilon = \frac{1}{2C n^{\beta}}$.  Thus
\(
  \frac{1}{\rho} - \epsilon \geqslant \frac{1}{2C n^{\beta}}
= \Omega(n^{-\beta})
\).
This leads to the lower bound
\(
  r \geqslant 2^{\Omega(n^{1 - 2 \beta})}
\)
as claimed in part \ref{item:2}. 
\end{proof}
\end{thm}

\section{Polyhedral Inapproximability of CLIQUE and SDPs}
\label{sec:consequences}

We will now use Theorem~\ref{thm:highNNRPert} in combination 
with Theorem~\ref{thm:framework} to lower bound the sizes of
approximate EFs for CLIQUE and some SDPs.
First, we pinpoint a pair $P,Q$ of nested
polyhedra that will be the source of our polyhedral inapproximability
results. Second, we give a faithful linear encoding of CLIQUE and 
prove strong lower bounds on the sizes of approximate EFs for 
CLIQUE w.r.t.\ this encoding. Third, we focus on approximations
of SDPs by LPs.

\subsection{A Hard Pair}
\label{sec:hard_pair}

Let $n$ be a positive integer. The \emph{correlation polytope} 
$\COR(n)$ is defined as the convex hull of all the $n \times n$ 
rank-$1$ binary matrices of the form $bb^\intercal$ where $b \in \{0,1\}^n$. 
In other words, 
\[
\COR(n)  = \conv{\set{bb^\intercal}{b \in \binSet^n}}.
\]
This will be our inner polytope $P$. Next, let  
\[
Q = Q(n) \coloneqq \{x \in \RR^{n \times n} \mid \inp{2\diag(a)-aa^\intercal}{x} \leqslant 1,\ a \in \{0,1\}^n\},
\]
where $\inp{\cdot}{\cdot}$ denotes the Frobenius inner product. 
This will be our outer polyhedron $Q$. 

Then the following is known, see \citep{extform4}. First, 
$P \subseteq Q$. Second, denoting by $S^{P,Q}$ the slack 
matrix of the pair $P,Q$, we have
\(
S^{P,Q}_{ab} = (1-a^\intercal b)^2
\). 
Thus, for $\rho \geqslant 1$, we have $S^{P,\rho Q}_{ab} 
= (1-a^\intercal b)^2 + \rho - 1$. Observe that the matrix 
\(S^{P,\rho Q}\) is a \(\rho\)-extension of UDISJ and therefore 
has high nonnegative rank via Theorem~\ref{thm:highNNRPert}; moreover
it has positive entries everywhere for \(\rho > 1\).
Together with Theorem~\ref{thm:sandwich} this implies that every
polyhedron sandwiched between $P = \COR(n)$ and $\rho Q$ has large extension complexity. We obtain the following theorem.

\begin{thm}[Lower bounds for approximate EFs of the hard pair]
\label{thm:corInapprox}
Let $\rho \geqslant 1$, let $n$ be a positive integer and let
$P = \COR(n)$, $Q = Q(n)$ be as above.  Then the following hold:
\begin{enumerate}
\item\label{item:01}
If \(\rho\) is a fixed constant, then \(\xc(P,\rho Q) = 2^{\Omega(n)}\).
\item\label{item:02}
If \(\rho = O(n^\beta)\) for some constant \(\beta < 1/2\), then
$\xc(P,\rho Q) = 2^{\Omega(n^{1 - 2 \beta})}$.
\end{enumerate}
\end{thm}

\subsection{Polyhedral Inapproximability of CLIQUE}
\label{sec:CLIQUE}

We define a 
natural linear encoding for the maximum clique problem (CLIQUE)
as follows.
Let $n$ denote the number of vertices of the input graph.
We define a $d = n^2$ dimensional encoding.
The variables are denoted by
$x_{ij}$ for $i, j \in [n]$. Thus $x \in \RR^{n \times n}$. The 
interpretation is that a set of vertices $X$ is encoded by
$x_{ij} = 1$ if $i, j \in X$ and $x_{ij} = 0$ otherwise.
Note that \(X = \{ i : x_{ii} = 1\}\)
can be recovered from only the diagonal variables.
This defines the set $\mathcal{L} \subseteq \{0,1\}^*$ of feasible
solutions. Notice that $x \in \{0,1\}^{n \times n}$ is feasible if 
and only if it is of the form $x = bb^\intercal$ for some 
$b \in \{0,1\}^n$, the characteristic vector of \(X\).
Thus we have $P = \COR(n)$
for the inner polytope.

The admissible objective functions are chosen as follows to encode
the CLIQUE problem for graphs \(G\) supported on \([n]\).
Given a graph $G$ such that $V(G) \subseteq [n]$, we let 
$w_{ii} \coloneqq 1$ for $i \in V(G)$, $w_{ii} \coloneqq 0$ for $i \in 
[n] \setminus V(G)$, $w_{ij} = w_{ji} \coloneqq -1$ when $ij$ is 
a non-edge of $G$ (that is, $i, j \in V(G)$, $i \neq j$ and $ij \notin E(G)$), and $w_{ij} = w_{ji} \coloneqq 0$ otherwise. 
We denote the resulting weight vector by $w^G$. Notice 
that for a graph $G$ with $V(G) = [n]$, we have $w^G = 
I - A(\overline{G})$ where $I$ is the $n \times n$ identity matrix, 
$A(\overline{G})$ is the adjacency matrix of the complement of $G$.

A feasible solution $x = bb^\intercal \in \{0,1\}^{n \times n}$ 
maximizes $\inp{w^G}{x}$ only if $b$ is the characteristic vector 
(or incidence vector) of a clique of $G$. Indeed, if $b = \chi^X$ 
and $ij$ is a non-edge of $G$ with $i, j \in X$ then removing $i$ 
or $j$ from $X$ increases $\inp{w^G}{x}$. Moreover, the maximum of 
$\inp{w^G}{x}$ over $x \in \{0,1\}^{n \times n}$ feasible is the clique 
number $\omega(G)$.

The admissible objective functions are
the ones of the form \(w^{G}\), i.e.,
$\mathcal{O} = \{w^{G} : V(G) \subseteq [n]\}$
is the set of admissible functions.
Therefore, $(\mathcal{L},\mathcal{O})$ defines a valid
linear encoding of CLIQUE. We denote the outer convex set of this linear encoding 
by $Q^{\mathrm{all}}$. It is actually the polyhedron defined as
$Q^{\mathrm{all}} = \{x \in \RR^{n \times n} \mid \forall$ graphs 
$G$ s.t.\ $V(G) \subseteq [n] : \inp{w^G}{x} \leqslant \omega(G),\
\forall i \neq j \in [n] : x_{ij} \geqslant 0\}$. We will now show
that \(Q^{\mathrm{all}} \subseteq Q\).

\begin{lem}
  Let \(Q^{\mathrm{all}},Q\) be as above, then \(Q^{\mathrm{all}}
  \subseteq Q\).
  \begin{proof}
Let $x \in Q^{\mathrm{all}}$. We want to prove that $x$ satisfies 
all the constraints defining $Q$. We show this by restricting
to graphs $G$ with $\omega(G) = 1$. For a given $a \in \{0,1\}^n$, 
let \(G\) be the graph with $\chi^{V(G)} = a$ and $E(G) = \emptyset$. Then,
\[
\inp{2\diag(a) - aa^T}{x} 
= \inp{w^G}{x} \leqslant \omega(G) = 1.
\]
The lemma follows.
\end{proof}
\end{lem}

Because $Q^{\mathrm{all}}$ is contained in the polyhedron $Q$ 
defined above, every $K$ satisfying $P \subseteq K \subseteq 
\rho Q^{\mathrm{all}}$ also satisfies
$P \subseteq K \subseteq \rho Q$.
Hence, Theorem~\ref{thm:corInapprox} yields the following 
result. 

\begin{thm}[Polyhedral inapproximability of CLIQUE]
\label{thm:CLIQUE}
W.r.t.\ the linear encoding defined above,
CLIQUE has an $O(n^2)$-size $n$-approximate EF.
Moreover, every $n^{1/2-\epsilon}$-approximate EF of CLIQUE
has size $2^{\Omega(n^{2\epsilon})}$, for all $0 < \epsilon < 1/2$.
\begin{proof}
The $n$-approximate EF of CLIQUE is trivial: it is defined by the 
system $\mathbf{0} \leqslant x \leqslant \mathbf{1}$, or in slack form 
$x - y = \mathbf{0}$, $x + z = \mathbf{1}$, $y \geqslant 
\mathbf{0}$, $z \geqslant \mathbf{0}$. We claim that this
defines a $n$-approximate EF of CLIQUE of size $2n^2$. Indeed,
letting $K = [0,1]^{n \times n}$ denote the polytope defined by this EF,
we have $P \subseteq K$. Moreover, $\max \{\inp{w}{x} \mid x \in K\}
\leqslant n \leqslant n \cdot \max \{\inp{w}{x} \mid x \in P\}$ for all 
admissible objective functions $w$ of dimension $n \times n$ with a nonzero
diagonal. In case an admissible $w$ has $w_{ii} = 0$ for all $i \in [n]$, 
we have $\max \{\inp{w}{x} \mid x \in K\} = 0 = \max \{\inp{w}{x} \mid x \in P\}$. Our claim and the first part of the theorem follows. 

The second part of the theorem follows directly from 
Theorem \ref{thm:corInapprox} and the fact that $Q^{\mathrm{all}} 
\subseteq Q$.
\end{proof}
\end{thm}

\subsection{Polyhedral Inapproximability of SDPs}
\label{sec:sandw-spectr}

In this section we show that there exists a spectrahedron
with small semidefinite extension complexity
but high approximate extension complexity;
i.e., any sufficiently fine polyhedral approximation is large.
This indicates that in general it is not 
possible to approximate SDPs arbitrarily well using small LPs, so 
that SDPs are indeed a much stronger class of optimization
problems. (The situation looks quite different for SOCPs, see 
\cite{Ben-TalNemirovski01}.) The result follows from 
Theorem~\ref{thm:corInapprox} and \cite{extform4}. 

We denote the vector space of all $r \times r$ symmetric matrices
by $\SYM^r$ and the cone of all $r \times r$ symmetric positive 
semidefinite matrices (shortly, the PSD cone) by $\SDP^r$.
A \emph{semidefinite EF} of a convex set $S \subseteq \RR^d$ is
a linear system $\inp{E_i}{x} + \inp{F_i}{Y} = g_i$ ($i \in [k]$), 
$Y \in \SDP^r$ where $E_i \in \RR^d$ and $F_i \in \SYM^r$, such 
that $x \in S$ if and only if $\exists Y \in \SDP^r$ with 
$\inp{E_i}{x} + \inp{F_i}{Y} = g_i$ for all $i \in [m]$. 
Thus a convex set admits a semidefinite EF if and only
if it is a spectrahedron. The \emph{size} of the semidefinite 
EF $\inp{E_i}{x} + \inp{F_i}{Y} = g_i$ ($i \in [k]$), $Y \in \SDP^r$ 
is simply $r$. The \emph{semidefinite extension complexity} 
of a spectrahedron $S \subseteq \RR^d$ is the minimum size 
of a semidefinite EF of $S$. This is denoted by $\xcp(S)$. 
A \emph{rank-\(r\) PSD-factorization} of a nonnegative matrix 
\(M \in \R^{m \times n}\) is given by matrices
\(T_1, \ldots, T_m \in \SDP^r\) and \(U^1, \ldots, U^n \in \SDP^r\),
so that \(M_{ij} =\inp{T_i}{U^j}\);
the \emph{PSD-rank of \(M\)} is the smallest \(r\) such that
there exists such a factorization.
Yannakakis's factorization theorem can be generalized to the
SDP-case (see \cite{GouveiaParriloThomas2011}), i.e., the 
semidefinite extension complexity of a pair of polyhedra $P, Q$
equals the PSD-rank of any associated slack matrix, in most 
cases (e.g., \cite{GRT13} prove the 
equality under the hypothesis that $Q$ does not contain 
any line).

Let $P = \COR(n)$ be the correlation polytope and
\(Q = Q(n) \subseteq \R^{n \times n}\) be the polyhedron 
defined above in Section~\ref{sec:hard_pair}. Although every 
polyhedron $K$ sandwiched between $P$ and $Q$ has 
super-polynomial extension complexity, and by 
Theorem~\ref{thm:corInapprox} this even applies 
to polyhedra sandwiched between $P$ and $\rho Q$ 
for $\rho = O(n^{1/2-\epsilon})$, there exists a spectrahedron
\(S\) sandwiched between $P$ and $Q$ with small semidefinite 
extension complexity. 

\begin{lem}[Existence of spectrahedron]
\label{lem:sandw-spectr}
Let $n$ be a positive integer and let $P  = \COR(n)$, $Q = Q(n)$ 
be as above. Then there exists a spectrahedron $S$ in 
$\RR^{n \times n}$ with $P \subseteq S \subseteq Q$ 
and $\xcp(S) \leqslant n+1$.
\begin{proof}
For \(a,b \in \binSet^{n}\), the matrices \(T^{a},U_b \in \SDP^{n+1}\) 
defined in \eqref{eq:PSD_factorization} satisfy $\inp{T^{a}}{U_{b}} 
= (1 - a^\intercal b)^2$. Let $M = M(n) \in \RR^{2^n \times 2^n}$
be the matrix defined as $M_{ab} = (1 - a^\intercal b)^2$. 
The matrix $M$ is an $O(n^2)$-rank nonnegative matrix
extending the UDISJ
matrix, and also the slack matrix of the pair $P$, $Q$. 
Then $M_{ab} = \inp{T_a}{U^b}$ is a rank-$(n+1)$ PSD-factorization 
of $M$. 

Consider the system 
\begin{equation}
\label{eq:semidefinite_EF}
\inp{2\diag(a)-aa^\intercal}{x} + \inp{T_a}{Y} = 1\quad 
(a \in \{0,1\}^n), \qquad Y \in \SDP^{n+1}
\end{equation}
and $S$ be the projection to $\RR^{n \times n}$ of the 
pairs $(x,Y) \in \RR^{n \times n} \times \SYM^{n+1}$ 
satisfying \eqref{eq:semidefinite_EF}. 

First observe that \(S \subseteq Q\): since \(T_{a} \in 
\SDP^{n+1}\) for all \(a \in \binSet^n\) and $Y \in 
\SDP^{n+1}$ we have \(\inp{T_a}{Y} \geqslant 0\) 
and thus \(\inp{2\diag(a)-aa^\intercal}{x} \leqslant 1\) 
holds for all \(x \in S\). 

In order to show that \(P \subseteq S\) recall that $M$ is the
slack matrix of the pair $P$, $Q$. Therefore, for each vertex 
$x \coloneqq bb^\intercal$ of $P$, we can pick \(Y \coloneqq U^{b}\)
from the factorization such that $\inp{2\diag(a)-aa^\intercal}{x} + 
\inp{T_a}{Y} = 1$ and $Y \in \SDP^{n+1}$. It follows that
\(P \subseteq S\). 
\end{proof}
\end{lem}

Our final result is the following inapproximability 
theorem for spectrahedra.
Let us denote the closed \(\varepsilon\)-neighbourhood of \(S\)
in the \(\ell_{1}\)-norm by
\(S^{\varepsilon} \coloneqq
\set{x \in \R^{n \times n}}{\exists x_{0} \in S \colon
  \onorm{x - x_{0}} \leqslant \varepsilon}\).

\begin{thm}[Polyhedral inapproximability of SDPs]
\label{thm:SDPConeInapprox}
Let $\rho \geqslant 1$, and let $n$ be a positive integer.
Then there exists a spectrahedron \(S \subseteq \R^{n \times n}\) 
with \(\xcp(S) \leqslant n+1\) such that for every polyhedron $K$ with
$S \subseteq K \subseteq S^{\rho - 1}$ the following hold:
\begin{enumerate}
\item\label{item:11}
If \(\rho\) is a fixed constant, then \(\xc(K) = 2^{\Omega(n)}\).
\item\label{item:12}
If \(\rho = O(n^\beta)\) for some constant \(\beta < 1/2\), then
$\xc(K) = 2^{\Omega(n^{1 - 2 \beta})}$.
\end{enumerate}
\begin{proof}
By Lemma~\ref{lem:sandw-spectr}, there is a spectrahedron \(S\) with
\(P \subseteq S \subseteq Q\) and \(\xcp(S) \leqslant
n+1\).
We now show \(S^{\rho - 1} \subseteq \rho Q\).
Let \(x \in S^{\rho - 1}\), and let \(x_{0} \in S\)
with \(\onorm{x - x_{0}} \leq \rho - 1\).
As \(S \subseteq Q\), we also have \(x_{0} \in Q\),
hence for every \(a  \in \{0,1\}^{n}\) we obtain
\begin{equation*}
 \begin{split}
  \inp{2\diag(a)-aa^\intercal}{x}
  &
  =
  \inp{2\diag(a)-aa^\intercal}{x - x_{0}} +
  \inp{2\diag(a)-aa^\intercal}{x_{0}}
  \\
  &
  \leq
  \underbrace{\maxNorm{2\diag(a)-aa^\intercal}}_{\leq 1}
  \cdot \underbrace{\onorm{x - x_{0}}}_{\leq \rho - 1} + 1
  \leq \rho.
 \end{split}
\end{equation*}
Therefore \(x \in \rho Q\).
Therefore, \(S^{\rho - 1} \subseteq \rho Q\) for $\rho
\geqslant 1$. If now \(K\) is a polyhedron such that \(S \subseteq 
K \subseteq S^{\rho - 1}\)
then also \(P \subseteq K \subseteq \rho Q\).
The result thus follows from Theorem~\ref{thm:corInapprox}.
\end{proof}
\end{thm}

\section{Concluding Remarks}

We have introduced a general framework to study 
approximation limits of small LP relaxations. Given 
a polyhedron \(Q\) encoding admissible objective 
functions and a polytope \(P\) encoding feasible 
solutions, we have proved that any LP relaxation 
sandwiched between \(P\) and a dilate \(\rho Q\) 
has extension complexity at least the nonnegative 
rank of the slack matrix of the pair $P$, $\rho Q$.

This yields a lower bound depending \emph{only}
on the linear encoding of the problem at hand, and 
applies \emph{independently} of the structure of the actual 
relaxation. By doing so, we obtain unconditional lower 
bounds on integrality gaps for small LP relaxations, which 
hold even in the unlikely event that $P = NP$. 

We have proved that every polynomial-size LP relaxation for 
(a natural linear encoding of) CLIQUE has essentially an
$\Omega(\sqrt{n})$ integrality gap. As mentioned above, this 
was recently improved by \cite{BM13} to a tight 
$\Omega(n^{1-\epsilon})$ integrality gap, see 
\cite{BP2013commInfo} for a short proof and many
generalizations.

Finally, our work sheds more light on the inherent 
limitations of LPs in the context of combinatorial optimization 
and approximation algorithms, in particular, in comparison
to SDPs. We provide strong evidence that certain 
approximation guarantees can only be achieved via 
non-LP-based techniques (e.g., SDP-based or combinatorial).

Actually, our work has inspired \cite{CLRS13} to prove 
lower bounds on the size of LPs for approximating Max~CUT,
Max~$k$-SAT and in fact any Max~CSP. Among other results, 
they obtain a $n^{\Omega(\log n / \log \log n)}$ lower 
bound on the size of any $(2-\epsilon)$-approximate EF 
for Max~CUT (of course, with nonnegative weights). 
\cite{CLRS13} thus proving the following conjecture
on Max~CUT that we stated in an earlier version of 
this text: 
\begin{thm}\cite{CLRS13}
It is not possible
to approximate Max~CUT with LPs of poly-size within 
a factor better than \(2\).
\end{thm}
This is in stark contrast
with the ratio achieved by the SDP-based algorithm of \cite{GoemansWilliamson95} which is known to be optimal, 
assuming the Unique Games Conjecture~\cite{Khot02a,%
KhotKMO04, MosselOO05}.

Next, about CLIQUE itself, here is an interesting question 
that this paper leaves open, as pointed out by one of the 
referees: find an $n$-vertex graph $G$ for which the clique
polytope 
\(
\CLIQUE(G) \coloneqq \conv{\{\chi^K \in \RR^{V(G)} \mid K \subseteq
V(G) \text{ is a clique}\}}
\) 
has no polynomial-size $n^{1-\epsilon}$-approximate EF. Note that
encoding CLIQUE through the clique polytope does not satisfy our
faithfulness condition.

Finally, so far no strong lower bounding technique for semidefinite 
EFs are known. It is plausible that in the near future we 
will see lower bounding techniques on the PSD rank that would 
be suited for studying approximation limits of SDPs. (We remark 
however that such bounds should not only argue on the zero/nonzero 
pattern of a slack matrix.)

\section{Acknowledgements}
  
We would like to thank the two referees for their time and comments
which contributed to improve the text.  
  
\bibliographystyle{abbrvnat}
\bibliography{../bibliography}

\end{document}